\documentclass[aap]{imsart}

\usepackage{amsmath} \usepackage{amsfonts}
\usepackage{amssymb} \usepackage{latexsym} \usepackage{enumerate}
\usepackage{mathtools}
\mathtoolsset{showonlyrefs}

 \usepackage{mathrsfs}
\usepackage[numbers]{natbib}

\input{style.tex}
\newcommand{\Xup}{X^{\uparrow}}
\newcommand{\Xdo}{X^{\downarrow}}
\renewcommand{\hat}{\widehat}

\begin{document}
\begin{frontmatter}

\title{Continuous-time Duality for Super-replication with Transient Price Impact}
 \runtitle{Duality with Transient Price Impact}

 \author{Peter Bank$^*$ and Yan Dolinsky$^\dagger$ \corref{}}

 \address{P. Bank: Department of Mathematics, TU Berlin.
   e.mail: bank@math.tu-berlin.de\\
   Y. Dolinsky Department of Statistics, Hebrew University and School
   of Mathematical \\Sciences, Monash University.
   e.mail:yan.dolinsky@mail.huji.ac.il} \affiliation{Hebrew
   University$^\dagger$ and Monash University$^\dagger$, and TU
   Berlin$^*$} \runauthor{P. Bank and Y. Dolinsky} 

\thankstext{}{The
   author YD is partially supported by the ISF Grant 160/17. PB wishes
   to thank Monash University for inviting him to visit for two weeks in 2018 when part of
   the research on this paper was completed.} 

\date{\today}

\maketitle
\begin{abstract}
   We establish a super-replication duality in a continuous-time
  financial model as in~\cite{BankVoss:18b} where an investor's trades
  adversely affect bid- and ask-prices for a risky asset and where
  market resilience drives the resulting spread back towards zero at
  an exponential rate. Similar to the literature on models with a
  constant spread (cf., e.g., \cite{CvitKarat:96,
    KallsenMuhleKarbe:10, CzichowskySchachermayer:17}), our dual
  description of super-replication prices involves the construction of
  suitable absolutely continuous measures with martingales close to
  the unaffected reference price. A novel feature in our duality is a
  liquidity weighted $L^2$-norm that enters as a measurement of this
  closeness and that accounts for strategy dependent spreads. As
  applications, we establish optimality of buy-and-hold strategies for
  the super-replication of call options and we prove a verification
  theorem for utility maximizing investment strategies.
\end{abstract}
\end{frontmatter}

\section{Introduction}\label{sec:1}

Financial models with transaction costs have been a great source of
intriguing challenges for stochastic analysis and control
theory. Starting with \cite{DavisNorm:90}, \cite{ShrSoner:94},
\cite{CvitKarat:96} strong emphasis has been put on the singular
control problems that emerge in models with a constant spread. The
duality theory for these models is now developed in great detail (see
\cite{JK,CvitKarat:96,Kaban,KabanStr:01,CampiSchachermayer:06,GRS10}). This has been used to study utility maximization
via its relation to shadow prices
(\cite{KallsenMuhleKarbe:10,GMKS2013,CzichowskySchachermayer:16,CzichowskySchachermayer:17,BY2019}) and has also been instrumental in the
development of asymptotic approaches for small transaction costs
(\cite{Schachermayer:17} and the references therein).

While convenient mathematically, the assumption of a constant spread
is justified only for very liquid assets. Less liquid assets will have
a spread that widens when a large transaction is being executed and,
upon completion of the transaction, the spread will decrease again due
to market resilience. This is well-known in the order execution
literature (\cite{ObiWang:13}, \cite{GatheralSchiedSlynko:12}) where
one derives optimal schedules for unwinding large positions that
account for such (at least partially) transient price
impact.

Following the approach proposed in \cite{BankVoss:18b}, we introduce a
model with transient price impact that allows for impact from both
buying and selling a risky asset. We even allow for stochastic market
depth and resilience that merely have to satisfy a certain
monotonicity assumption required to obtain convex wealth
dynamics. Instead of the utility maximization problem of interest
in~\cite{BankVoss:18b}, we focus here on the fundamental problem of
super-replicating an arbitrary contingent claim in a cost optimal
way. For models with a constant spread, the duality theory of this
problem is well-understood in terms of consistent price systems that
are based on the construction of measures with martingales that do not
deviate from the asset price by more than the exogenously given
spread; see \cite{Sch:14} and the reference therein. This
structure is recovered here but, due to the endogenous nature of our
spreads, we also have to optimally determine these. Our main result,
Theorem~\ref{thm:1}, shows how to suitably penalize possible choices
by a liquidity-dependent $L^2$-norm, characterizing the
super-replication costs in the form of a convex risk measure. An
interesting point to observe is that, contrary to the models with
exogenous spread, our model does not require any notion of
admissibility for our trading strategies. As already observed in a
model with purely temporary price impact in \cite{GuasoniRasonyi:15},
this is due to the impossibility to scale strategies at will since
such scaling incurs super-linearly growing costs.

The proof of this result rests on a particularly convenient expression
for the terminal wealth resulting from a strategy that also reveals
the convexity of this functional. As usual, a lower bound on
super-replication costs is comparably easy to obtain given the consistent price
system structure imposed by our dual variables. The proof of absence
of a duality gap, i.e., establishing an upper bound is more
involved. The first step is a rather standard separation argument
(Lemma~\ref{lem:1}) which gives us a suitable pricing measure. As a
second step, we introduce the martingale of the consistent price
system as a Lagrange multiplier enforcing the terminal liquidation
constraint (Lemma~\ref{lem:2}). The crucial third step is the
construction of a suitable spread and the identification of its
liquidity dependent $L^2$-norm as the correct penalty term for our
duality (Lemma~\ref{lem:3}). This is made possible by applying a
stochastic representation theorem from~\cite{BankElKaroui:04} which so
far was used only in connection with one-sided singular control
problems (\cite{BankRiedel:01}, \cite{ChiarollaFerrari:14},
\cite{Ferrari:15}, \cite{BankKauppila:17}) and here finds its first
application in a two-sided control problem with bounded variation
rather than increasing controls.

As an application, we show that also in our transient price impact
model the best way to super-replicate a call option is, under natural
conditions, to buy and hold the asset until maturity. This is in line
with results on models with exogenous spread;
cf. \cite{SSC95,Kus95,LS97,BT00,JLR03,GRS:08}.  We also provide a
verification result for identifying utility maximizing strategies by
the construction of suitable shadow prices similar to results with
fixed spread (\cite{KallsenMuhleKarbe:10,CzichowskySchachermayer:16})
and to a result with purely temporary price impact
(\cite{GuasoniRasonyi:15}).

\section{Trading in a transient price impact model}
\label{sec:2}

We consider a financial model with transient price impact similar to
\cite{BankVoss:18b}. Specifically, we consider a ``large'' investor who
can invest in a riskless savings account bearing zero interest (for
simplicity) and whose trades into and out of a risky asset move bid-
and ask prices that, in addition, are also driven by some exogenous
noise. This noise will be specified by a continuous, adapted process
$P=(P_t) _{t\geq 0}$ on a filtered probability space
$(\Omega, (\mathcal F_t)_{t\geq 0}, \mathbb P)$ where $\cF_0$ is
generated by the $\P$-null sets. We will assume that the filtration is
continuous:

\begin{Assumption}\label{asp:1}
  All $(\mathcal F_t)_{t\geq 0}$-martingales have a continuous
  version.
\end{Assumption}

\begin{Remark}
  This assumption is satisfied, e.g., if $(\mathcal F_t)_{t\geq 0} $
  is generated by a Brownian motion. It rules out complete surprises
  as generated, for instance, by the jumps of Poisson processes. The
  assumption ensures that there will not be any common jumps of
  trading policies and our martingale prices to be introduced later
  and it will allow us to apply a stochastic representation theorem
  from \cite{BankElKaroui:04} which is key for our analysis. From the
  duality theory of proportional transaction costs, see in particular
  \cite{CzichowskySchachermayer:16}, it is known that exogenous jumps
  lead to the need for l\`adl\`ag strategies and a considerably more
  delicate analysis which, in our context for strategy-dependent
  spreads, we have to leave for future research. Moreover, jumps, also
  by $P$ or the market depth process $\delta$ to be introduced
  shortly, would pose the challenge to specify what information on the
  jump is available when to the investor and how he can act on
  it. While certainly relevant from a financial-economic point of
  view, these questions are also beyond the scope of the present
  paper. 
\end{Remark}

The large investor's trading strategy is described by his given
initial holdings $x_0 \in \RR$ and a right-continuous, predictable
process $X=(X_t)_{t \geq 0}$ of bounded variation specifying the
number of risky assets held at any time. We denote by $\Xup$ and
$\Xdo$ the right-continuous predictable increasing and decreasing part
resulting from the Hahn-decomposition of
\begin{align}
  \label{eq:1}
  X_t=x_0+\Xup_t-\Xdo_t, \; t \geq 0, \quad X_{0-}\set x_0, \quad \Xup_{0-}\set \Xdo_{0-}\set0.
\end{align}
The set of all such strategies will be denoted by $\cX$.

Trades will permanently affect the mid-price $P^X$ which, in
line with \cite{HubermanStanzl:04}, we let take the linear form
\begin{align} \label{eq:2}
  P^X_t \set P_t + \iota X_t, \; t\geq 0, \quad P^X_{0-} \set P_0 + \iota x_0,
\end{align}
for some impact parameter $\iota \geq 0$. Trades will in addition drive bid- and
ask-prices away from the mid-price. Without further interventions,
market resilience lets bid- and ask-prices then gradually revert
towards the mid-price. We model this by letting the half-spread
follow the dynamics
\begin{align}
  \label{eq:3}
  d\zeta^X_t =  \frac{1}{\delta_t}(d\Xup_t+d\Xdo_t)-r_t \zeta^X_t dt,
  \quad \zeta^X_{0-} \set \zeta_0,
\end{align}
for a given initial value $\zeta_0 \geq 0$ and a given market depth
process $\delta$ and resilience rate $r$.

\begin{Remark} One way to interpret these spread dynamics is to think
  of trades eating into their respective side of the limit order book,
  widening the spread to an extent which depends on the current order
  book height $\delta_t$ while the market's resilience ensures that,
  without further trades, the spread will diminish at the exponential
  rate $r_t$. For simplicity, we assume the order book height at any
  time to be constant across ticks and identical for the ask- and the
  bid-side. More flexible nonlinear spread dynamics as
  in~\cite{AlfFruSch:10, PredShaShr:11} are conceivable but beyond the
  scope of the present paper. Note also that the mesoscopic time-scale
  underlying our model does not allow us to accommodate all the market
  microstructure effects so crucial for high-frequency trading, but
  instead suggests to view our model's market depth and resilience
  processes also as mesoscopic specifications of these market
  characteristics that would in practice need to be calibrated, e.g.,
  to moving averages of order book heights and order flow dynamics
  accounting for both limit and market orders; see
  \cite{ContKukanovStoikov:14} for an empirical study in this vein
  that also supports linear price impact specifications as in our
  stylized model.
\end{Remark}

We will require the following regularity of market depth $\delta$ and
resilience rate $r$:

\begin{Assumption}\label{asp:2}
  The market depth $\delta=(\delta_t)_{t\geq 0}>0$ is continuous and
  adapted. The resilience rate $r=(r_t) _{t\geq 0} \geq 0$ is predictable and such that
  $\delta/\rho$ is uniformly bounded away from zero and infinity on
  $\Omega \times [0,T]$ for any finite time horizon $T<\infty$ where
  \begin{align}
    \label{eq:4}
    \rho_t \set \exp\left(\int_0^t r_s \,ds\right), \quad t \geq 0.
  \end{align}
  Moreover, the resilience rate dominates the changes in market depth
  in the sense that
  \begin{align}
  \label{eq:5}
    \kappa_t \set \delta_t/\rho^2_t \text{ is strictly decreasing in }  t \geq 0.
  \end{align}

\end{Assumption}
\begin{Remark}
  As will become apparent in  Lemma~\ref{lem:1}, condition~\eqref{eq:5} is
  needed to ensure that the wealth dynamics are convex. When $\delta$
  is absolutely continuous it ammounts to the requirement
  \begin{align}
    \label{eq:6}
   \frac{1}{2} \frac{d}{dt} \log \delta_t < r_t, \quad t \geq 0,
  \end{align}
  i.e. relative changes in the market's depth have to be dominated by
  the market's resilience. In particular, Assumption~\ref{asp:2} holds
  when $\delta$ and $r$ are strictly positive constants. The question
  whether one can develop a duality theory without this assumption is
  left for future research; see, however, \cite{BankFruth:14} for
  considerations in this direction in a deterministic order execution
  problem.
\end{Remark}

By time $T \in (0,\infty)$ the induced investor's cash position will have evolved
from its given initial value $\xi_{0} \in \RR$ to $\xi^X_T$ as determined
by the profits and losses made from trading in and out of the risky
asset. These trades are executed half the spread away from the
mid-price $P^X$ and so the terminal cash position is
\begin{align}\label{eq:7}
\xi_T^X
& \set \xi_{0}-\int_{[0,T]} P^X_t \circ dX_t - \int_{[0,T]}
  \zeta^X_t \circ d(\Xup_t+\Xdo_t).
\end{align}

\begin{Remark}
  The above $\circ$-integrals are understood in the sense that for two
  RCLL processes $X,Y$ with $X$ of bounded variation, we let
  \begin{align}
    \label{eq:8}
    \int_{[0,T]} Y_t \circ dX_t \set \int_{[0,T]}
    \frac{1}{2}(Y_{t-}+Y_{t+}) dX_t,
  \end{align}
  where on the right-hand side we have a standard Lebesgue-integral
  with respect to the signed measure $dX$.  In~\eqref{eq:7}, this way
  of integrating accounts for the fact that, when buying assets in a
  bulk $\Delta \Xup_t>0$, both the mid-pricee $P^X$ and the
  half-spread $\zeta^X$ will increase only gradually during the order
  execution, letting the investor effectively trade at the average
  between pre- and post-transaction mid-price and the average between
  pre- and post-transaction half spread. We refer to~\cite{ObiWang:13}
  for similar considerations in an order execution
  framework. Alternatively, it is possible to consider
  $\int Y \circ dX$ as a Marcus integral for our controlled
  system. For our linear impact specification this amounts to the
  Stratonovich-like integral~\eqref{eq:7};
  cf. \cite{BechererBilarevFrentrup:17} and the references therein.
\end{Remark}

\section{Duality for super-replication of contingent claims}

Let us now consider the classical super-replication problem for a
cash-settled European contingent claim with $\cF_T$-measurable payoff
$H \geq 0$ at time $T \geq 0$ that is \emph{not} affected by the large
investor's trades, but exogenously given, for instance, as a
functional of the given unaffected price process $P$.
We will give a
dual description of such an exogenous payoff's super-replication costs
\begin{align}
  \label{eq:9}
  \pi(H) \set \inf\{\xi_{0} \in \RR \;:\; \xi^X_T \geq H \text{ for some
   } X \in \cX \text{ with } X_T=0\}.
\end{align}
\begin{Remark}
  We have to confine ourselves to claims whose payoff are not affected
  by the large investor because we have to preserve the convexity of
  the super-replication problem. Of course, pricing and hedging claims
  with payoffs that can be affected (or even manipulated) by the large
  investor is a practically (and in the aftermath possibly judicially)
  most relevant problem. But this would require a rather
  product-specific analysis and is thus beyond the scope of this
  duality paper. See, however, the PDE approaches in, e.g.,
  \cite{BouchardLoeperZou:17, BechererBilarev:18} as well as
  \cite{Jarr:94, BankBaum:04} for some results in this direction. Note
  also that for a vanilla option, whose payoff only depends on the
  terminal mid-price at time~$T$, the liquidation constraint $X_T=0$
  ensures that $P^X_T=P_T$ and thus prevents any manipulation
  possibilities, making our duality result below applicable to these
  products.
\end{Remark}

On the dual side of our description of super-replication costs, the
market frictions will be captured by the optional random measure $\mu$
that, under Assumption~\ref{asp:2}, is induced by the continuous
increasing process $-\kappa=-\delta/\rho^2$ on $(0,T)$ with point mass
$\kappa_T= \delta_T/\rho^2_T$ in $T$:
\begin{align}
  \label{eq:10}
  \mu(dt) \set 1_{(0,T)}(t) |d \kappa_t| + \kappa_T \mathrm{Dirac}_T(dt).
\end{align}
With this notation, we can now formulate our main result:
\begin{Theorem}\label{thm:1}
  Under Assumptions~\ref{asp:1} and~\ref{asp:2}, the super-replication
  costs~\eqref{eq:9} of a contingent claim $H \geq 0$ have the dual
  description
\begin{align}
  \label{eq:11}
  \mbox{}\!\!\!\!
\pi(H) = \sup_{(\QQ,M,\alpha)} \left\{\E_{\QQ}[H] - \frac{1}{2}
  \|\alpha-\zeta_0\|^2_{L^2(\QQ \otimes
  \mu)}-M_0x_0-\frac{1}{2}\iota x_0^2\right\}>-\infty
\end{align}
where the supremum is taken over all triples $(\QQ,M,\alpha)$ of
probability measures $\QQ \ll \PP$ on $\cF_T$, martingales
$M \in \cM^2(\QQ)$ and all optional
$\alpha \in L^2(\QQ \otimes \mu)$ which control the fluctuations of
$P$ in the sense that
\begin{align}
  \label{eq:12}
  |P_t-M_t| \leq \frac{\rho_t}{\delta_t} \E_\QQ\left[\int_{[t,T]} \alpha_u \,
  \mu(du)\middle|\cF_t\right], \quad 0 \leq  t \leq T.
\end{align}
\end{Theorem}

\subsection{Comparison with other super-replication duality formulae}

Let us discuss the above duality formula for super-replication prices
by comparing it with other such dualities obtained in different
financial models.

First, the supremum on the right-hand side of~\eqref{eq:11} includes
all measures $\QQ \ll \P$ for which $P$ is a square-integrable
martingale (if there are any). For these one can choose $M=P$ and
$\alpha=\zeta_0$ to satisfy the constraint~\eqref{eq:12} and obtain
that $\pi(H) \geq \E_{\QQ}[H]-x_0P_0$ when ignoring permanent impact
($\iota=0$). This inequality is clearly in line with the classical
frictionless super-replication duality. (Notice that the value of the
initial position $x_0P_0$ is subtracted here because
$\pi(H)$ in~\eqref{eq:10} describes the super-replication costs in
cash required when starting with a position of $x_0$ in the risky
asset.)

Let us next turn to models with transaction costs arising from a fixed
spread. Adjusting the multiplicative settings considered in
\cite{CvitKarat:96, CampiSchachermayer:06, CzichowskySchachermayer:16} to an
additive one as considered here leads  to consistent price
systems given by $\P$-martingales $(Z^0,Z^1)$ with $Z^0_0=1$,
$Z^0_T>0$ such that the distance of $M \set Z^1/Z^0$ to $P$ is
dominated by the (constant for simplicity) half-spread $\lambda$ which
one has to pay on top of $P$ when buying and which is subtracted from
the proceeds when selling a unit of the risky asset. One can then
define $\QQ$ by $d\QQ/d\PP \set Z^0_T$ and put $\alpha \set \lambda$
to obtain a triple $(\QQ,M,\alpha)$ as required by our duality
formula, e.g., in any model with zero resilience ($r=0$, $\rho=1$) and initial
spread $\zeta_0 = \lambda$. Indeed, \eqref{eq:12} does hold for any
market depth $\delta>0$ (which has to be decreasing to meet
Assumption~\ref{asp:2}) since then
\begin{align}
  \label{eq:13}
  \frac{\rho_t}{\delta_t} \E_\QQ\left[\int_{[t,T]} \alpha_u \,
  \mu(du)\middle|\cF_t\right] = \frac{\rho_t}{\delta_t} \lambda
  \E_\QQ\left[\mu([t,T])\middle|\cF_t\right] = \frac{\rho_t}{\delta_t}
  \lambda \kappa_t = \lambda.
\end{align}
As a result, ignoring possible permanent impact ($\iota=0$),
$\pi(H) \geq \E_{\QQ}[H]-Z^1_0x_0$, in line with the super-replication
results for models with fixed spread.

Observe that, contrary to these models, our setting with spread impact
does not require any notion of admissibility for trading
strategies. Also, in our model we have, regardless of the initial
position $x_0$, that $\pi(0)>-\infty$ for \emph{any} choice of
(continuous) price process $P$. Hence, even for specifications
allowing for the most egregious arbitrage in a fixed-spread model (let
alone in a frictionless one), there is no way to reach zero terminal
wealth from arbitrarily low initial cash positions. This is due to the
fact that scaling favorable strategies ultimately turns these
unfavorable as transaction costs effectively grow quadratically when
scaling a strategy, not just linearly as in any setting with a fixed
spread. This effect has been observed in an Almgren-Chriss
\cite{AlmgrenChriss:01}-style model with temporary rather than
transient market impact in \cite{GuasoniRasonyi:15}.  Like our
super-replication cost formula, theirs takes the form of a convex risk
measure rather than a coherent one as found for the fixed spread
models. This is again due to the nonlinear scaling of transaction
costs.


\subsection{Applications}

To illustrate the usefulness of the above duality result let us derive
in this section the super-replication costs of a call option and show
how to verify optimality of a proposed investment strategy.

\subsubsection{Super-replicating call options}\label{sec:CallOptions}

As a first application of our super-replication duality, let us verify
that also in our model with strategy-depen\-dent spread, buy-and-hold is
the best way to super-replicate a call option
\begin{align}
  \label{eq:14}
   H = (P_T-k)^+ \quad \text{with} \quad k \geq 0,
\end{align}
at least if liquidity coefficients are deterministic and if the
unaffected price $P$ satisfies the conditional full-support property
(see \cite{GRS:08})
\begin{align}
  \label{eq:15}
  \supp \P\left[(P_u)_{t \leq u \leq T} \in \cdot \middle|\cF_t\right] = C_{P_t}([t,T],\RR_+),
  \quad 0 \leq t \leq T,
\end{align}
where, for $p \geq 0$, $C_p([t,T],\RR_+)$ denotes the class of
continuous, nonnegative functions $f$ on $[t,T]$ with $f(t)=p$.

\begin{Corollary}\label{cor:1}
  Let Assumption~\ref{asp:1} hold true and let market depth and
  resilience be deterministic and satisfy Assumption~\ref{asp:2}. In
  addition, suppose $P$ is strictly positive with the conditional full
  support property~\eqref{eq:15}.  Then, for an investor with initial
  position $x_0 \leq 1$, the super-replication cost of a cash-settled
  call option is
 \begin{align}
   \label{eq:16}
    \pi((P_T-k)^+) &= P_0(1-x_0) -\frac{1}{2}\iota x^2_0+\zeta_0(1-x_0)+
\frac{(1-x_0)^2}{2\delta_0}\\&\qquad+\frac{\zeta_0 +(1-x_0)/\delta_0}{\rho_T}+\frac{1}{2\delta_T}
 \end{align}
 and it is attained by holding one unit of the risky asset over
 $[0,T)$ to be sold at time $T$.
\end{Corollary}
\begin{proof}
  Let us consider the strategy that immediately takes its position in
  the risky asset to one unit and keeps it there until unwinding it in
  the end:
  \begin{align}
    \label{eq:17}
    \hat{X}^\uparrow \set (1-x_0)  1_{[0,T]}, \hat{X}^\downarrow \set
   1_{\{T\}}, \hat{X}=1_{[0,T)}.
  \end{align}
  When starting with the cash position $\xi_0$ given by the right-hand
  side of~\eqref{eq:16} this leads by~\eqref{eq:7} to the terminal wealth
  \begin{align}
    \label{eq:18}
    \xi^{\hat{X}}_T = P_T \geq (P_T-k)^+ = H.
  \end{align}
  Here, the estimate holds true as $P$ is nonnegative. So the
  right-hand side of~\eqref{eq:16} is sufficient initial cash to
  super-replicate the call.

  We will use our duality formula from Theorem~\ref{thm:1} to show
  that $\epsilon>0$ less than this amount is not sufficient. To this end, we choose
  \begin{align}
    \label{eq:19}
    \alpha_t \set
    \zeta_0+\frac{1-x_0}{\delta_0}+\frac{\rho_T}{\delta_T}1_{\{T\}}(t),
    \quad 0 \leq t  \leq T,.
  \end{align}
  Clearly, there exists a Lipschitz continuous, non-increasing deterministic function
  $g:[0,T]\rightarrow\mathbb R$ with
  \begin{align}
    \label{eq:20}
  g_0&=\frac{\int_{[0,T]} \alpha_u \,  \mu(du)}{\delta_0},  \\
  g_T&=-\frac{\alpha_T}{\rho_T},\\
 |g_t|&\leq\frac{\rho_t}{\delta_t}\int_{[t,T]} \alpha_u \mu(du),
    \quad 0 \leq t \leq T.
  \end{align}
  Lemma~\ref{lem.new} below yields a probability measure
  $\QQ\ll\mathbb P$ with $\mathbb Q(P_T>\epsilon)< \epsilon$ and a
  square integrable $\QQ$-martingale $M$ such that
  \begin{align}
  &|g_t+P_t-M_t|<\epsilon\inf_{0\leq t\leq T}\frac{\rho_t}{\delta_t}\mu([t,T]), \  0 \leq t\leq T.
  \end{align}
  Hence, the triple $(\mathbb Q,M,\alpha+\epsilon)$ is as
  requested by our Duality  Theorem~\ref{thm:1}. Using the simple inequality
  $P_T\leq (P_T-k)^{+}+\epsilon+k1_{\{P_T>\epsilon\}}$, we thus obtain
 \begin{align}
 \pi(H)\geq& \mathbb E_{\mathbb Q}[P_T-M_T]+M_0-\mathbb E_{\mathbb Q}[P_T-(P_T-k)^{+}]\\
 &\quad -\frac{1}{2}  \|\alpha-\zeta_0\|^2_{L^2(\QQ \otimes  \mu)}-M_0x_0-\frac{1}{2}\iota x_0^2\\
  \geq&\frac{\alpha_T}{\rho_T}+(1-x_0)P_0+(1-x_0)\frac{\int_{[0,T]} \alpha_u \,
  \mu(du)}{\delta_0}\\
  &\quad -\frac{1}{2}\int_{[0,T]} |\alpha_u-\zeta_0|^2d\mu(u)-
  \frac{1}{2}\iota x_0^2-O(\epsilon).
 \end{align}
The result follows by using $\mu([0,T])=\delta_0$ and taking $\epsilon\downarrow 0$.
\end{proof}
\begin{Lemma}\label{lem.new}
  Suppose $P>0$ exhibits the conditional full support
  property~\eqref{eq:15} and let $g:[0,T]\rightarrow\mathbb R$ be
  Lipschitz-continuous and non-increasing. Then, for any $\epsilon>0$,
  there is a probability measure $\QQ \ll \PP$ and a
  square-integrable $\QQ$-martingale $M$ such that $\QQ$-almost surely
 \begin{align}
   \label{eq:21}
   |g_t+P_t-M_t|<\epsilon,  \quad  0\leq t\leq T,
 \end{align}
 and
 \begin{align}
   \label{eq:21+}
   \mathbb Q(P_T>\epsilon)< \epsilon.
 \end{align}
\end{Lemma}
\begin{proof}
  Without loss of generality, we can assume that $T=1$ and fix
  $0<\epsilon<1$. It will be convenient to denote increments of a
  given process $(X_t)$ by
  $\Delta^N_n X \set X_{\frac{n}{N}}-X_{\frac{n-1}{N}}$ where
  $N \in \NN$ and $n=1,\dots,N$. For such $n, N$ and for $\sigma>0$,
  consider the disjoint events $A^{N,\sigma}_{+,n}$ and
  $A^{N,\sigma}_{-,n}$ given by
\begin{align}
  \label{eq:22}
  A^{N,\sigma}_{\pm,n} \set & \left\{\Delta^N_{n} (P+g) =\pm
                             (P_{\frac{n-1}{N}} \wedge N^{1/4}) \frac{\sigma}{\sqrt{N}}+o
                              \text{ for some } o
                              \in [0,1/N^2]\right\}\\&\quad \cap
                              \left\{\max_{\frac{n-1}{N} \leq t \leq
                                                       \frac{n}{N}}|P_t-P_{\frac{n-1}{N}}|
                                                       \leq \epsilon/3\right\}.
\end{align}
For $N>(6\sigma/\epsilon)^4$ (as assumed henceforth), the path
properties described for $P$ in the definition of both
$A^{N,\sigma}_{+,n}$ and $A^{N,\sigma}_{-,n}$ are met by non-empty
open subsets of $C_{P_{\frac{n-1}{N}}}([\frac{n-1}{N},1],\RR_+)$ for
any given $P_{\frac{n-1}N}>0$. Here, insisting on nonnegative paths is
possible because $g$ is assumed to be non-increasing (whence
$\Delta^N_n g \leq 0$). It thus follows from the conditional full support
property~\eqref{eq:15} that
 \begin{align}
   \label{eq:23}
   \P\left[A^{N,\sigma}_{+,n} \middle| \cF_{\frac{n-1}{N}}\right]>0
   \quad \text{and} \quad    \P\left[A^{N,\sigma}_{-,n,} \middle|
   \cF_{\frac{n-1}{N}}\right]>0, \quad n=1,\dots,N.
 \end{align}
 So  there is $\QQ^{N,\sigma} \ll \PP$ for which 
 \begin{align}
   \label{eq:24}
   \QQ^{N,\sigma}\left[A^{N,\sigma}_{+,n} \middle| \cF_{\frac{n-1}{N}}\right]=
   \QQ^{N,\sigma}\left[A^{N,\sigma}_{-,n} \middle| \cF_{\frac{n-1}{N}}\right] =
   \frac{1}{2}, \quad n=1,\dots,N;
 \end{align}
 for instance $\QQ^{N,\sigma}$ with density
\begin{align}
  \label{eq:25}
  \frac{d\QQ^{N,\sigma}}{d\PP} \set \prod_{n=1,\dots,N} \frac{1}{2}\left(\frac{1_{A^{N,\sigma}_{+,n}}}{\P\left[A^{N,\sigma}_{+,n} \middle| \cF_{\frac{n-1}{N}}\right]}+\frac{1_{A^{N,\sigma}_{-,n}}}{\P\left[A^{N,\sigma}_{-,n} \middle| \cF_{\frac{n-1}{N}}\right]}\right)
\end{align}
will do. In conjunction with the definition of $A^{N,\sigma}_{\pm,n}$,
this ensures that $\QQ^{N,\sigma}$-a.s.
 \begin{align}
   \label{eq:26}
   \left|\E_{\QQ^{N,\sigma}}\left[\Delta^N_{n}(P+g)\middle|\cF_{\frac{n-1}{N}}\right]\right|
   \leq \frac{1}{N^2}, \quad n=1,\dots,N,
 \end{align}
 and, thus, the auxiliary discrete-time martingale
 \begin{align}
   \label{eq:27}
   \tilde{M}_n \set P_0+g_0+\sum_{m=1}^n
   \left(P_{\frac{m}{N}}-\E_{\QQ^{N,\sigma}}\left[P_{\frac{m}{N}}\middle|\cF_{\frac{m-1}{N}}\right]\right),
   \quad n=0,\dots,N,
 \end{align}
 satisfies $\QQ^{N,\sigma}$-a.s.
 \begin{align}
   \label{eq:28}
   \left|P_{\frac{n}{N}}+g_{\frac{n}{N}}-\tilde{M}_n\right| \leq
   \frac{1}{N}, \quad n=0,\dots,N.
 \end{align}
 Combining this with the Lipschitz-continuity of $g$ and the
 $\epsilon/3$-bound on the fluctuations of $P$ over any time interval
 of length $\frac{1}{N}$  from the definition of $A^{N,\sigma}_{\pm,n}$
 yields
   \begin{align}
     \label{eq:29}
     \left|\tilde{M}_n-\tilde{M}_{n-1}\right|\leq \frac{\epsilon}{2},
     \quad n=1,\dots,N, \quad \QQ^{N,\sigma}\text{-a.s.}
   \end{align}
for $N>N_0$, where $N_0(\sigma)$ depends only on $\sigma$, $\epsilon$, and the Lipschitz
constant~$L$ of~$g$.

We conclude that the bounded $\QQ^{N,\sigma}$-martingale given by
\begin{align}
  \label{eq:30}
  M^{N,\sigma}_t \set
  \E_{\QQ^{N,\sigma}}\left[\tilde{M}_N\middle|\cF_{t}\right], \quad 0
  \leq t \leq T,
\end{align}
 satisfies $M^{N,\sigma}_{\frac{n}{N}}=\tilde M_n$, $n=0,\dots,N$,  and
 \begin{align}
   \label{eq:31}
 \max_{n=1,\dots,N}\max_{\frac{n-1}{N}\leq t\leq\frac{n}{N}}|M^{N,\sigma}_t-M^{N,\sigma}_{\frac{n-1}{N}}|
\leq \frac{\epsilon}{2} \quad \QQ^{N,\sigma}\text{-a.s.}  
 \end{align}
 This together with~\eqref{eq:28} and the $\epsilon/3$-bound on the
 fluctuations of $P$ from the definition of $A^{N,\sigma}_{\pm,n}$
 gives that $g+P-M^{N,\sigma}$ satisfies the required bound~\eqref{eq:21}
 $\QQ^{N,\sigma}$-a.s.\  for $N>N_0(\sigma)$.

 It remains to argue that $\sigma$ and then $N>N_0(\sigma)$ can be
 chosen such that $\QQ\set\QQ^{N,\sigma}$ from the above construction
 also satisfies the second requirement
 $\QQ(P_1>\epsilon)<\epsilon$. To this end, note that the difference
 equation
 \begin{align}
   \label{eq:32}
   Z^{N,\sigma}_0 &\set P_0, \\ 
     \Delta^N_{n} Z^{N,\sigma} &\set (Z^{N,\sigma}_{\frac{n-1}{N}}
                                 \wedge N^{1/4})
                                          \frac{\sigma}{\sqrt{N}}\left(1_{A^{N,\sigma}_{+,n}}-1_{A^{N,\sigma}_{-,n}}\right)+\frac{L+1}{N},
                                          \quad n=1,\dots,N,
 \end{align}
 yields a process $Z^{N,\sigma}$ dominating $P$ in the sense that
 $Z^{N,\sigma}_{\frac{n}{N}} \geq P_{\frac{n}{N}}$, $n=0,\dots,N$,
 $\QQ^{N,\sigma}$-a.s., as follows readily by induction using the
 definition of $A^{N,\sigma}_{\pm,n}$ and the Lipschitz continuity of
 $g$. Theorem~4.4 in \cite{DuffieProtter:92} in conjunction
 with~\eqref{eq:28} yields that, as $N \uparrow \infty$, the
 distribution of $Z^{N,\sigma}_1$ under $\QQ^{N,\sigma}$ converges to
 the distribution of $Z^{(\sigma)}_1$ where $Z^{(\sigma)}$ is the
 (unique) solution of the linear SDE
 $$ 
 Z^{(\sigma)}_0=P_0, \quad dZ^{(\sigma)}_t= Z^{(\sigma)}_t \sigma
 dW_t+(L+1) dt
 $$
 for some standard Brownian motion $W$.  In view of~\eqref{eq:32}, we
 can thus choose $\sigma$ and $N>N_0(\sigma)$ to fulfill the requirement
 $\QQ^{N,\sigma}(P_1>\epsilon)<\epsilon$ provided that $Z^{(\sigma)}_1$ converges
 to $0$ in probability as $\sigma\uparrow\infty$. For this, observe
 that
\begin{align}
Z^{(\sigma)}_1&=P_0 e^{\sigma W_1-\sigma^2/2}\left(1+\int_{0}^1 (L+1)e^{-\sigma W_t+\sigma^2 t/2}dt \right)\\
&\leq P_0 e^{\sigma W_1-\sigma^2/2}+P_0 (L+1) e^{-\sigma^2/(2\ln\sigma)}
\int_{0}^{1-1/\ln {\sigma}} e^{\sigma (W_1-W_t)} dt\\
&\quad+P_0 (L+1) \int_{1-1/\ln {\sigma}}^{1} e^{\sigma (W_1-W_t)-\sigma^2 (1-t)/2}dt.
\end{align}
Clearly, the first two summands in the last expression vanish almost
surely while, due to Fubini's theorem, the expectation of the last one is
$$
\E\left(\int_{1-1/\ln {\sigma}}^{1} e^{\sigma (W_1-W_t)-\sigma^2 (1-t)/2}dt\right)=
\frac{1}{\ln\sigma} \to 0
$$
for $\sigma \uparrow \infty$. This shows that indeed $\lim_{\sigma\uparrow\infty}Z^{(\sigma)}_1=0$ in probability and the proof is completed.

\end{proof}

\subsubsection{Utility maximization by duality}\label{sec:VerificationTheorem}

Super-replication duality is often used to study utility maximization
problems which, in turn, allow for less conservative and practically
more useful contingent claim valuation paradigms such as indifference
pricing. While this paper has to leave indifference valuation for
future research, let us note here a verification theorem to illustrate
the suitability of our duality concepts for this theory:

\begin{Corollary}\label{cor:2}
  Let Assumptions~\ref{asp:1} and~\ref{asp:2} hold true and consider a
  strictly concave, increasing and differentiable utility function $u$
  for which
 $$
 \sup_{X \in \cX \text{ with } X_T=0} \E[u(\xi^X_T)\vee 0]<\infty.
 $$
 Suppose $\hat{X}\in \cX$ with $\hat{X}_T=0$ yields via
 \begin{align}
   \label{eq:33}
   \frac{d\hat{\QQ}}{d\P} \set \frac{u'(\xi^{\hat{X}}_T)}{\E[u'(\xi^{\hat{X}}_T)]}
 \end{align}
 a probability measure $\hat{\QQ} \ll \P$ which allows for a
 \emph{shadow price} $\hat{M}$ for spread dynamics
 \begin{align}
   \label{eq:34}
   \hat{\lambda}_t \set \frac{\rho_t}{\delta_t} \E_{\hat{\QQ}} \left[
   \int_{[t,T]} \hat{\alpha}_u \mu(du) \middle| \cF_t\right], \quad 0
   \leq t \leq T,
 \end{align}
  with $\hat{\alpha} \set \rho \zeta^{\hat{X}} \in
  L^2(\hat{\QQ}\otimes \mu)$, i.e., for a $\hat{\QQ}$-square
  integrable martingale $\hat{M}$ such that
  \begin{align}
    \label{eq:35}
    P_t - \hat{\lambda}_t \leq \hat{M}_t \leq P+\hat{\lambda}_t, \quad
    0 \leq t \leq T,
  \end{align}
  with equality almost surely holding true in the first and second
  estimate on the support of $d\hat{X}^{\downarrow}$ and
  $d\hat{X}^{\uparrow}$, respectively.

  Then $\hat{X}$ yields the highest expected utility $\E[u(\xi^X_T)]$
  among all strategies $X \in \cX$ with $X_T=0$.
\end{Corollary}

The proof of this corollary will follow readily from considerations
required for the proof of Theorem~\ref{thm:1}. We thus postpone it to
the end of Section~\ref{sec:ProofLowerBound}. We adopted the notion of
shadow prices from the theory of optimal investment with proportional
transaction costs (see, e.g., \cite{CvitKarat:96,
  KallsenMuhleKarbe:10, CzichowskySchachermayer:17}) where the
martingales $\hat{M}$ with the stated flat-off conditions are
constructed explicitly or emerge from duality of utility
maximization. In our setting, the construction of shadow prices is
more challenging as the spread $\hat{\lambda}$ is not given
exogenously. It is thus not obvious how to construct optimal
investment policies $\hat{X}$ from the above verification result. See,
however, \cite{BankVoss:18b} for a convex analytic approach to
exponential utility maximization when $P$ is a Brownian motion with
drift and $\delta$ and $r$ are constant.









\section{Proof of the duality theorem}
\subsection{Preliminaries}

Let us prepare the proof of Theorem~\ref{thm:1} by rewriting the
profits and losses from trading in our price impact model:

\begin{Lemma}\label{lem:1}
  Suppose Assumption~\ref{asp:2} holds true. Then, for any strategy
  $X \in \cX$ with $X_T=0$, we have
\begin{align}
  \label{eq:36}
  \xi^X_T
& = v_0 -\Lambda^X_T
\end{align}
where
\begin{align}
  \label{eq:37}
  v_0 \set  \xi_{0}+\frac{1}{2}(\iota x_0^2+\delta_0 \zeta_{0}^2)
\end{align}
and
\begin{align}
  \label{eq:38}
  \Lambda^X_T \set \int_{[0,T]}  P_t \,dX_t +
  \frac{1}{2} \int_{[0,T]}  (\eta^X_t)^2 \, \mu(dt)
\end{align}
with
\begin{align}
  \label{eq:39}
  \eta^X_t \set \rho_t \zeta^X_t = \zeta_0+\int_{[0,t]} \frac{\rho_s}{\delta_s} \,d(\Xup_s+\Xdo_s), \quad 0 \leq t \leq T.
\end{align}
Moreover, there is a constant $C>0$, depending only on the bounds on
$\delta/\rho$ from Assumption~\ref{asp:2}, such that, for any
$X \in \cX$, we have
\begin{align}
  \label{eq:40}
  \Xup_T+\Xdo_T \leq C\left(l+\sup_{0 \leq t \leq T}
  |P_t|\right) \text{ on } \{\Lambda^X_T \leq l^2\}.
\end{align}
Finally, the mapping $X \mapsto \Lambda^X_T$ is convex and
lower-semicontinuous. More precisely, if $X^n \in \cX$ converges
weakly to $X \in \cX$ in the sense that almost surely $X^{n,\uparrow}$
and $X^{n,\downarrow}$ converge weakly as Borel-measures on $[0,T]$
to, respectively, some adapted, right-continuous, increasing $A$ and
$B$ with $X=x_0+A-B$, $A_{0-}=B_{0-}=0$, then almost surely
\begin{align}
  \label{eq:41}
  \liminf_n \Lambda^{X^n}_T \geq \Lambda^X_T.
\end{align}
\end{Lemma}
\begin{proof}
\begin{enumerate}
\item
Let us first prove our formula~\eqref{eq:36} for $\xi^X_T$. For the
integral of the mid-price we get by continuity of $P$ that
\begin{align}
 \int_{[0,T]} P^X_t \circ dX_t &= \int_{[0,T]} P_t \,dX_t +
  \iota \int_{[0,T]} X_t \circ dX_t \\&=  \int_{[0,T]} P_t \,dX_t +
  \iota\frac{1}{2}(X^2_T-x_0^2)  \label{eq:42}
\end{align}
where the last identity is due to the chain rule for Stratonovich
integrals. Similarly, using $\zeta^X = \eta^X/\rho$ and
$d(\Xup_t+\Xdo_t)=\frac{\delta_t}{\rho_t}d\eta^X_t$, we get
\begin{align}
 \int_{[0,T]} \zeta^X_t \circ d(\Xup_t+\Xdo_t)
 & =  \int_{[0,T]} \frac{\delta_t}{\rho_t^2} \eta^X_t\circ d\eta^X_t
= \int_{[0,T]} \kappa_t\circ
      d\left(\frac{1}{2}(\eta^X_t)^2\right)
\\&=
    \kappa_T\frac{1}{2}(\eta^X_T)^2-\delta_0\frac{1}{2}\zeta_0^2
    - \int_{(0,T)}  \frac{1}{2}(\eta^X_t)^2 \,d\kappa_t
\\&= \frac{1}{2}\int_{[0,T]} (\eta^X_t)^2 \,\mu(dt)-\frac{1}{2}\delta_0\zeta^2_0.  \label{eq:43}
 \end{align}
Combining~\eqref{eq:42} with~\eqref{eq:43} we obtain~\eqref{eq:36}
when $X_T=0$.

\item For $X \in \cX$, it follows from the definition~\eqref{eq:38} of
  $\Lambda^X_T$ that on $\{\Lambda^X_T \leq l^2\}$ we have
 \begin{align}
   \label{eq:44}
  l^2+\sup_{t \in [0,T]} |P_t|(\Xup_T+\Xdo_T)
&\geq l^2 - \int_{[0,T]} P_t \,dX_t
\\& \geq \frac{1}{2}\int_0^T(\eta^X_t)^2 \mu(dt)
\geq (\Xup_T+\Xdo_T)^2/C
 \end{align}
 for some constant $C>0$ only depending on the bounds on $\delta/\rho$
 from Asssumption~\ref{asp:2}. Hence, $x \set \Xup_T+\Xdo_T$ is such
 that $x^2 \leq C(px +l^2)$ for $p \set \sup_{t \in [0,T]} |P_t|$. This
 implies~\eqref{eq:40}.

\item Let $X_0, X_1 \in \cX$ and observe that then
  $\frac{1}{2}(\Xup_0+\Xup_1)-\frac{1}{2}(\Xdo_0+\Xdo_1)$
  is a decomposition of $X \set \frac{1}{2}(X_0+X_1)$ into the difference of
  two right-continuous increasing processes. It follows that
  $\frac{1}{2}(\Xup_0+\Xup_1)-\Xup$ and
  $\frac{1}{2}(\Xdo_0+\Xdo_1)-\Xdo$ are increasing and so $0 \leq \eta^X \leq
  \frac{1}{2}(\eta^{X_0}+\eta^{X_1})$. In light of~\eqref{eq:38}, this
  yields the convexity of $\Lambda^X$.

  Similarly, for $X^n$ converging to $X=x_0+A-B$ as described in the
  lemma, $A-\Xup$ and $B-\Xdo$ are increasing. Hence, we have
  $\eta^{X^n}_t \to \eta^{x_0+A+B}_t \geq \eta^X_t$ in $t=T$ and in
  every point of continuity $t$ for $A+B$. By continuity of $P$, we
  also have
  $$\lim_n \int_{[0,T]} P_t \,dX^n_t = \int_{[0,T]} P_t \, dX_t.$$
  So lower-semicontinuity of $X \mapsto \Lambda^X_T$ is a consequence
  of~\eqref{eq:38} and Fatou's lemma.
\end{enumerate}
\end{proof}

\subsection{Proof of the lower bound}
\label{sec:ProofLowerBound}

Observe first that the supremum in~\eqref{eq:11} is greater than
$-\infty$. Indeed we can take any $\QQ^0 \ll \P$ for which
$\alpha^0_t \set \sup_{0 \leq s \leq t} |P_s\rho_s|$, $0 \leq t \leq T$, is in
$L^2(\QQ^0\otimes\mu)$ and let $M^0 \set 0$ to obtain a triple
$(\QQ^0,M^0,\alpha^0)$ satisfying the constraint~\eqref{eq:12}. Indeed, we
then have
\begin{align*}
  \frac{\rho_t}{\delta_t} \E_{\QQ^0}\left[\int_{[t,T]} \alpha^0_u \,
  \mu(du)\middle|\cF_t\right]
  &\geq \frac{\rho_t}{\delta_t} \alpha^0_t
  \E_{\QQ^0}\left[\mu([t,T])\middle|\cF_t\right]
  = \frac{\alpha^0_t}{\rho_t}\\&\geq |P_t|= |P_t-M^0_t|, \quad 0 \leq
                                 t \leq T.
\end{align*}
Hence, the supremum in~\eqref{eq:11} cannot be $-\infty$.

To prove that it gives a lower bound, consider $\xi_0 \in \RR$ and
$X \in \cX$ with $X_T=0$ such that $\xi^X_T \geq H \geq 0$ and let
$(\QQ,M,S)$ be a triple as in Theorem~\ref{thm:1}.

\begin{Lemma}\label{lem:2}
 We have
\begin{align}
  \label{eq:46}
  \Xup_T+\Xdo_T, \sup_{0 \leq t \leq T} |P_t| \in L^2(\QQ).
\end{align}
\end{Lemma}
\begin{proof}
  By Doob's maximal inequality,
  $\sup_{t \in [0,T]} |M_t| \in L^2(\QQ)$. Similarly,
  $\alpha \in L^2(\QQ \otimes \mu)$ yields that also the supremum over
  $[0,T]$ of the right-hand side of~\eqref{eq:12} is in
  $L^2(\QQ)$. Together with our previous observation, this implies
  that also $\sup_{0 \leq t \leq T} |P_t| \in L^2(\QQ)$. Square-integrability of
  $\Xup_T+\Xdo_T$ is now immediate from~\eqref{eq:40} with
  $l^2 \set v_0 = \xi^X_T+\Lambda^X_T \geq \Lambda^X_T$ because
  $\xi^X_T \geq H \geq 0$ almost surely.
\end{proof}

By Lemma~\ref{lem:1}, the super-replication property of $X$ is
tantamount to
\begin{align}
  \label{eq:47}
  v_0 \geq H + \int_{[0,T]} P_t \,dX_t + \frac{1}{2} \int_{[0,T]}
  (\eta^X_t)^2 \,\mu(dt).
\end{align}
Observe that by~\eqref{eq:12} we can estimate
\begin{align}
  \int_{[0,T]} P_t \,dX_t & = \int_{[0,T]} (P_t-M_t)\,dX_t  + \int_{[0,T]} M_t \,dX_t \\
& \geq - \int_{[0,T]} |P_t-M_t| (d\Xup_t+d\Xdo_t)  -
  M_0x_0-\int_0^T X_t \,dM_t
\\& = - \int_{[0,T]} |P_t-M_t| \frac{\delta_t}{\rho_t} \,d\eta^X_t  -
  M_0x_0-\int_0^T X_t \,dM_t, \label{eq:48}
\end{align}
where we first used integration by parts and $X_T=0$ and then
that~\eqref{eq:3} gives
$d\eta^X_t = \rho_t /\delta_t(d\Xup_t+d\Xdo_t)$. Square-integrability
of $M$ and~\eqref{eq:46} yield
$\E_{\QQ} \left[\int_0^T X^2_t d[M]_t^{1/2}\right]<\infty$, ensuring that
$\int_0^. X_t dM_t$ is a true martingale. Hence, taking expectation
in~\eqref{eq:48} we find
\begin{align}
 \E_\QQ &\left[\int_{[0,T]} P_t \,dX_t \right]\\ \label{eq:49}
&\geq - \E_\QQ\left[\int_{[0,T]} |P_t-M_t| \frac{\delta_t}{\rho_t} \,d\eta^X_t  +
  M_0x_0\right]
\\& \geq -\E_\QQ\left[\int_{[0,T]} \E_\QQ\left[\int_{[t,T]} \alpha_u \, \mu(du) \middle| \cF_t\right] \,d\eta^X_t  +
  M_0x_0 \right] \label{eq:50}
\\&= -\E_\QQ\left[\int_{[0,T]} \int_{[0,u]} d\eta^X_t \alpha_u \, \mu(du)  +
  M_0x_0\right]
\\&= -\E_\QQ\left[\int_{[0,T]}(\eta^X_u- \zeta_0) \alpha_u \, \mu(du)  +
  M_0x_0\right],
\end{align}
where in the second estimate we used~\eqref{eq:12} and the first
identity follows from Fubini's theorem in conjunction with the
observation that the conditional expectation in~\eqref{eq:50} can be
dropped as it gives the optional projection of
$(\int_{[t,T]} \alpha_u \, \mu(du))_{0 \leq t \leq T}$.

Now we take expectation in~\eqref{eq:47} and use the preceding estimate
to obtain
\begin{align}
  \label{eq:51}
  v_0 &\geq \E_\QQ\left[H+\int_{[0,T]}\left\{\frac{1}{2} (\eta^X_t)^2 -(\eta^X_t- \zeta_0) \alpha_t\right\} \mu(dt)  -
  M_0x_0\right]
\\&=\E_\QQ\left[H+\int_{[0,T]}\left\{\frac{1}{2}(\eta^X_t -\alpha_t)^2-\frac{1}{2}(\alpha_t- \zeta_0)^2+\frac{1}{2}\zeta_0^2 \right\}\mu(dt)  -
  M_0x_0\right]
  \\ &\geq \E_\QQ\left[H+\int_{[0,T]}\left\{-\frac{1}{2}(\alpha_t- \zeta_0)^2+\frac{1}{2}\zeta_0^2 \right\}\mu(dt)  -
  M_0x_0\right]
\\&= \E_{\QQ}[H] - \frac{1}{2}
    \E_\QQ\left[\int_{[0,T]}(\alpha_t- \zeta_0)^2 \mu(dt)\right]
    + \frac{1}{2} \zeta_0^2 \delta_0-M_0x_0
\end{align}
where in the last step we used that
$\mu([0,T])=\kappa_0=\delta_0$. Recalling the definition~\eqref{eq:37}
of $v_0$, this gives
\begin{align}
  \label{eq:52}
  \xi_0 \geq \E_{\QQ}[H] - \frac{1}{2}\|\alpha-\zeta_0\|^2_{L^2(\QQ
  \otimes \mu)}-M_0x_0-\frac{1}{2}\iota x_0^2,
\end{align}
which yields the claimed lower bound.

It is at this point easy to also give the \textbf{proof of the
  verification result stated in Corollary~\ref{cor:2}}. For this, take
any $X \in \cX$ and note that, by concavity of $u$,
\begin{align}
  \label{eq:53}
  u(\xi^X_T)-u(\xi^{\hat{X}}_T) \leq u'(\xi^{\hat{X}}_T)(\xi^{{X}}_T-\xi^{\hat{X}}_T).
\end{align}
Taking expectations under $\P$ and recalling the definition of $\hat{\QQ}$, it
thus suffices to argue
\begin{align}
  \label{eq:54}
  \E_{\hat{\QQ}}[\xi^X_T] \leq   \E_{\hat{\QQ}}[\xi^{\hat{X}}_T].
\end{align}
For this, note that from~\eqref{eq:36} we have
\begin{align}
  \label{eq:55}
  \E_{\hat{\QQ}}[\xi^X_T] =
v_0-  \E_{\hat{\QQ}} \left[\int_{[0,T]} P_t \,dX_t \right]
-\frac{1}{2}\E_{\hat{\QQ}} \left[\frac{1}{2} \int_{[0,T]}  (\eta^{X}_t)^2 \, \mu(dt)\right].
\end{align}
Proceeding as for~\eqref{eq:49}, \eqref{eq:50}, we estimate
\begin{align}
  \label{eq:56}
   \E_{\hat{\QQ}} \left[\int_{[0,T]} P_t \,dX_t \right] \geq
-\E_{\hat{\QQ}}\left[\int_{[0,T]}(\eta^X_u- \zeta_0) \hat{\alpha}_u \, \mu(du)  +
  \hat{M}_0x_0\right]
\end{align}
and observe that for $X=\hat{X}$ we actually get an equality
here  due to the support assumption~\eqref{eq:35}. Therefore,
 \begin{align}
   \label{eq:57}
   \E_{\hat{\QQ}}[\xi^X_T] &\leq v_0+
\E_{\hat{\QQ}}\left[\int_{[0,T]}(\eta^X_t- \zeta_0) \hat{\alpha}_t  -\frac{1}{2}(\eta^X_t)^2\, \mu(dt)  +
  \hat{M}_0x_0\right]\\
&=v_0+\E_{\hat{\QQ}}\left[\int_{[0,T]} \frac{1}{2}\left\{
  (\hat{\alpha}_t  -\zeta_0)^2 -(\eta^X_t-\hat{\alpha}_t)^2-\zeta_0^2\right\}\mu(dt)+\hat{M}_0x_0\right]\\
&\leq v_0+\E_{\hat{\QQ}}\left[\int_{[0,T]} \frac{1}{2}\left\{
  (\hat{\alpha}_t-\zeta_0)^2-\zeta_0^2\right\}\mu(dt)+\hat{M}_0x_0\right]
 \end{align}
 where, again, we have equality everywhere for $X=\hat{X}$ by choice
 of $\hat{\alpha}=\eta^{\hat{X}}$. It follows that~\eqref{eq:54} does
 hold true as remained to be shown.

\subsection{Proof of the upper bound}

In order to prove ``$\leq$'' in our dual description~\eqref{eq:11}, we
have to construct for any $\hat{\xi}_0<\pi(H)$ a triple
$(\hat{\QQ},\hat{M},\hat{\alpha})$ as considered in Theorem~\ref{thm:1} such that
\begin{align}
  \label{eq:58}
  \hat{\xi}_0 < \E_{\hat{\QQ}}[H] - \frac{1}{2}\|\hat{\alpha}-\zeta_0\|^2_{L^2(\hat{\QQ}
  \otimes \mu)}-\hat{M}_0x_0-\frac{1}{2}\iota x_0^2.
\end{align}

Observe that, by changing to an equivalent measure if necessary, we
can assume without loss of generality that
\begin{align}
  \label{eq:59}
  H\in L^1(\PP), \sup_{0 \leq t \leq T}|P_t| \in L^6(\PP).
\end{align}
For notational convenience, let us introduce the class
\begin{align}
  \label{eq:60}
   \cX^2 \set \{X \in \cX \;:\; \Xup_T+\Xdo_T \in L^2(\PP)\}
\end{align}
and let us denote by
\begin{align}
  \label{eq:61}
  \hat{v}_0 \set  \hat{\xi}_{0}+\frac{1}{2}(\iota x_0^2+\delta_0 \zeta_{0}^2)
\end{align}
the constant from~\eqref{eq:37} corresponding to
$\xi_0 = \hat{\xi}_0$.

We start with the construction of $\hat{\QQ}$ which emerges from a
standard separation argument:

\begin{Lemma}\label{lem:3}
  There is a probability measure $\hat{\QQ}$ with bounded density with
  respect to $\PP$ such that
 \begin{align}
   \label{eq:62}
   \hat{v}_0<\E_{\hat{\QQ}}[H]+\inf_{X \in \cX^2\text{ with } X_T=0}
   \E_{\hat{\QQ}}\left[\Lambda^X_T\right].
 \end{align}
\end{Lemma}
\begin{proof}
In light of our expression~\eqref{eq:36} for the investor's terminal
cash position, the condition $\hat{\xi}_0<\pi(H)$ translates into
\begin{align}
  \label{eq:63}
  H-\hat{v}_0 \not\in \cC \set \left\{-\Lambda^X_T-A: X \in \cX \text{
  with}
 \, X_T=0, A \in L^0_+(\cF_T)\right\}.
\end{align}

We will argue below that $\cC$ is a convex and closed subset of
$L^0(\cF_T)$.  It follows then that $\cC \cap L^1(\P)$
is a convex and closed subset of $L^1(\P)$ that, by~\eqref{eq:63},
does not contain $H-\hat{v}_0 \in L^1(\P)$. By the Hahn-Banach
Separation Theorem we can thus find $Z \in L^\infty(\cF_T)-\{0\}$ such
that
\begin{align}
  \label{eq:64}
  \E[Z(H-\hat{v}_0)] > \sup_{C \in \cC \cap L^1(\P)} \E[ZC].
\end{align}
Since $L^1_-(\P)-\Lambda^0_T \subset \cC$, we must have $Z \geq 0$ almost
surely. We can therefore define a probability measure $\hat{\QQ} \ll \P$ via
\begin{align}
  \label{eq:65}
  \frac{d\hat{\QQ}}{d\P} \set \frac{Z}{\E[Z_T]}.
\end{align}
Then~\eqref{eq:64} readily yields~\eqref{eq:62} upon observing that
for $X \in \cX^2$ we have $\Lambda^X_T \in L^1(\P)$ due to
Assumption~\ref{asp:2} and~\eqref{eq:59}.

It remains to prove that $\cC$ is indeed a convex, closed subset of
$L^0(\cF_T)$. Convexity is immediate from the convexity of
$X \mapsto \Lambda^X_T$ established in Lemma~\ref{lem:1}. For
closedness take $X^n \in \cX$ with $X^n_T=0$ and
$A^n \in L^0_+(\cF_T)$, $n=1,2,\dots$, such that $\Lambda^{X^n}_T+A^n$
converges in $L^0(\P)$ or, without loss of generality, even almost
surely to some finite limit $L$. We have to show that $-L \in \cC$,
i.e.,
\begin{align}
  \label{eq:66}
  L \geq \Lambda^X_T \text{ for some } X \in \cX.
\end{align}
By the given convergence, $\sup_n \Lambda^{X^n}_T$ is finite almost
surely. Hence, by our estimate~\eqref{eq:40} also
$\sup_n(X^{n,\uparrow}_T+X^{n,\downarrow}_T)$ is finite almost
surely. In particular,
$\conv(X^{n,\uparrow}_T+X^{n,\downarrow}_T,n=1,2,\dots)$ is bounded
almost surely, and thus in probability. So, by a Komlos-lemma as
Lemma~3.4 of \cite{Guasoni:02} or Lemma~3.1 in \cite{BankVoss:18b},
there is a cofinal sequence of convex combinations $\tilde{X}^n$ of
$X^n, X^{n+1},\dots$, such that almost surely $\tilde{X}^{n,\uparrow}$
and $\tilde{X}^{n,\downarrow}$ converge weakly as Borel-measures on
$[0,T]$ to, respectively, $A$ and $B$, two adapted, right-continuous,
and increasing processes with $A_{0-}=B_{0-}=0$. By
lower-semicontinuity and convexity of $X \mapsto \Lambda^X_T$,
see~\eqref{eq:41} in Lemma~\ref{lem:1}, it follows that for
$X \set x_0+ A-B \in \cX$ we indeed have
\begin{align}
  \label{eq:67}
  \Lambda^X_T \leq \liminf_n \Lambda^{\tilde{X}^n}_T \leq \liminf_n \Lambda^{X^n}_T \leq L
\end{align}
as desired.
\end{proof}

The martingale $\hat{M}$ is constructed as a Lagrange multiplier for
the constraint $X_T=0$ in the infimum of~\eqref{eq:62}:

\begin{Lemma}\label{lem:4}
 We have
 \begin{align}
   \label{eq:68}
   \inf_{X \in \cX^2\text{ with } X_T=0}&
   \E_{\hat{\QQ}}\left[\Lambda^X_T\right]
=\sup_{M \in \cM^2(\hat{\QQ})}\inf_{X \in \cX^2}
   \E_{\hat{\QQ}}\left[\Lambda^X_T-M_TX_T\right].
 \end{align}
\end{Lemma}
In conjunction with~\eqref{eq:62}, this lemma shows in particular that
there is $\hat{M} \in \cM^2(\hat{\QQ})$ with
\begin{align}
   \label{eq:69}
\hat{v}_0 < \E_{\hat{\QQ}}[H] +\inf_{X \in \cX^2}
\E_{\hat{\QQ}}\left[\Lambda^X_T-\hat{M}_TX_T\right].
 \end{align}
\begin{proof}
 We start by observing that
  \begin{align}
    \label{eq:70}
   \inf_{X \in \cX^2\text{ with } X_T=0}
   \E_{\hat{\QQ}}\left[\Lambda^X_T\right]
& = \lim_n    \inf_{X \in \cX^2} \left\{
   \E_{\hat{\QQ}}\left[\Lambda^X_T\right]+n\|X_T\|_{L^2(\hat{\QQ})}\right\}
 \\&=  \lim_n \inf_{X \in \cX^2}
  \sup_{\|M_T\|_{L^2(\hat{\QQ})} \leq n} \E_{\hat{\QQ}}\left[\Lambda^X_T-M_TX_T\right].
  \end{align}
  Indeed, the second identity is immediate as is ``$\geq$'' in the
  first line. For ``$\leq$'' there, take $X^n \in \cX^2$ such that
  $\E_{\hat{\QQ}}\left[\Lambda^{X^n}_T\right]+n\|X^n_T\|_{L^2(\hat{\QQ})}$
  approaches the limit in the first line.  Then
  $\sup_n \E_{\hat{\QQ}}\left[\Lambda^{X^n}_T\right]<\infty$ and, by
  convexity of $X \mapsto \Lambda^X_T$, we even have
  $\sup_{X \in \conv(X^n, n=1,2,\dots)}
  \E_{\hat{\QQ}}\left[\Lambda^{X}_T\right]<\infty$. It thus follows
  from~\eqref{eq:40} that
  $\conv(X^{n,\uparrow}_T+X^{n,\downarrow}_T, n=1,2,\dots)$ is bounded
  in $L^2(\hat{\QQ})$. In particular, it is bounded in $L^0$ and we
  can thus apply a Komlos-result such as Lemma~3.1
  in~\cite{BankVoss:18b} to obtain
  $\tilde{X}^n \in \conv(X^n,X^{n+1},\dots)$, $n=1,2,\dots$, that
  converge to some $\tilde{X} \in \cX$ in the way required for the
  lower-semicontinuity statement~\eqref{eq:41} in
  Lemma~\ref{lem:1}. We claim that 
  \begin{align}
    \label{eq:71}
    \tilde{X}_T=0 \text{ with } \E_{\hat{\QQ}} \left[\Lambda^{\tilde{X}}_T\right] \leq \lim_n \left\{
  \E_{\hat{\QQ}}\left[\Lambda^{\tilde{X}^n}_T\right]+n\|\tilde{X}^{n}_T\|_{L^2(\hat{\QQ})}\right\}.
  \end{align}
  Then, since by construction of the $(\tilde{X}^n)_{n=1,2,\dots}$
  this limit coincides with the one in~\eqref{eq:70}, we obtain that
  ``$\leq$'' must hold there. For the proof of~\eqref{eq:71} note that
  $(\Lambda^{\tilde{X}^n}_T)$ is bounded in $L^1(\hat{\QQ})$ because
  $\conv(X^{n,\uparrow}_T+X^{n,\downarrow}_T, n=1,2,\dots)$ is bounded
  in $L^2(\hat{\QQ})$. With the limit in~\eqref{eq:71} finite, this
  implies $\|\tilde{X}^{n}_T\|_{L^2(\hat{\QQ})} \to 0$ and so indeed
  $\tilde{X}_T=0$. For the estimate in~\eqref{eq:71}, observe that by
  Fatou's lemma and the lower-semicontinuity of $X \mapsto \Lambda^X$,
  it thus suffices to show that
  $(\Lambda^{\tilde{X}^n}_T \wedge 0)_{n=1,2,\dots}$ is uniformly
  $\hat{\QQ}$-integrable. This, in turn, follows by observing that due
  to H\"older's inequality (with $p=4$, $q=4/3$)
  \begin{align}
    \label{eq:72}
    \E_{\hat{\QQ}}\left[\left|\Lambda^{\tilde{X}^n}_T \wedge 0\right|^{3/2}\right]
&\leq \E_{\hat{\QQ}}\left[\left|\int_{[0,T]} P_t
  \,d\tilde{X}^n_t \wedge 0\right|^{3/2}\right]
\\& \leq \E_{\hat{\QQ}}\left[\sup_{0 \leq t \leq T}|P_t|^{3/2}(\tilde{X}^{n,\uparrow}_T+\tilde{X}^{n,\downarrow}_T)^{3/2}\right]
\\& \leq \E_{\hat{\QQ}}\left[\sup_{0 \leq t \leq T}|P_t|^{6}\right]^{1/4}\E_{\hat{\QQ}}\left[(\tilde{X}^{n,\uparrow}_T+\tilde{X}^{n,\downarrow}_T)^2\right]^{3/4}
  \end{align}
  is bounded because of~\eqref{eq:59} and because of the already
  established $L^2(\hat{\QQ})$-boundedness of
  $\conv(X^{n,\uparrow}_T+X^{n,\downarrow}_T, n=1,2,\dots)$.

  With~\eqref{eq:70} established, we obtain our assertion~\eqref{eq:68}
  from the minimax relation
  \begin{align}
    \label{eq:73}
    \inf_{X \in \cX^2}&
    \sup_{\|M_T\|_{L^2(\hat{\QQ})} \leq n}
    \E_{\hat{\QQ}}\left[\Lambda^X_T-M_TX_T\right]\\& =
    \sup_{\|M_T\|_{L^2(\hat{\QQ})} \leq n} \inf_{X \in \cX^2}\E_{\hat{\QQ}}\left[\Lambda^X_T-M_TX_T\right].
  \end{align}
  For this we endow $\cX^2$ with the $L^2(\P)$-norm of the
  $\omega$-wise total variation of its elements,
  $\|X\| \set \E_{\P}[(\Xup_T+\Xdo_T)^2]^{1/2}$, and the
  $L^2(\hat{\QQ})$-ball with the weak topology. Then both of these
  sets are convex subsets of topological vector spaces and the latter
  set is even compact. Moreover,
  $(X,M_T) \mapsto \E_{\hat{\QQ}}\left[\Lambda^X_T-M_TX_T\right]$ is
  continuous and convex in $X$ and continuous and concave (even
  affine) in $M_T$. We can thus apply Sion's minimax theorem
  (\cite{Komiya:88}) to obtain~\eqref{eq:73}.
 \end{proof}

Our final lemma constructs $\hat{\alpha}$:
\begin{Lemma}\label{lem:5}
  There is an optional $\hat{\alpha} \in L^2(\hat{\QQ} \otimes \mu)$ such that
  \begin{align}
    \label{eq:74}
    |P_t-\hat{M}_t| \leq \frac{\rho_t}{\delta_t}
    \E_{\hat{\QQ}}\left[\int_{[t,T]} \hat{\alpha}_u \,\mu(du)\middle | \cF_t
    \right],
 \quad 0 \leq t \leq T,
  \end{align}
 and
 \begin{align}
   \label{eq:75}
\inf_{X \in \cX^2} &
   \E_{\hat{\QQ}}\left[\Lambda^X_T-\hat{M}_TX_T\right]
 = -\frac{1}{2}\|\hat{\alpha}-\zeta_0\|^2_{L^2(\hat{\QQ}\otimes \mu)}-\hat{M}_0x_0+\frac{1}{2}\zeta_0^2\delta_0.
 \end{align}
  \end{Lemma}
\begin{proof}
 We first use integration by parts along with the observation that
$\E_{\hat{\QQ}}[\int_0^T X^2_t \,d[\hat{M}]^{1/2}_t]<\infty$ for
$X\in\cX^2$ to obtain that for such $X$ we can write
 \begin{align}
    \label{eq:76}
\E_{\hat{\QQ}}&\left[\Lambda^X_T-\hat{M}_TX_T\right]
\\&=\E_{\hat{\QQ}}\left[\int_{[0,T]} (P_t-\hat{M}_t) \,dX_t+\frac{1}{2}\int_{[0,T]}
(\eta^X_t)^2 \mu(dt)-\hat{M}_0x_0\right]
\\& = \E_{\hat{\QQ}}\left[-\int_{[0,T]} |P_t-\hat{M}_t| \,d\tilde{X}_t+\frac{1}{2}\int_{[0,T]}
(\eta^{\tilde{X}}_t)^2 \mu(dt) \right]-\hat{M}_0x_0
  \end{align}
where $\tilde{X}_t \set x_0-\int_{[0,t]}
\sign{(P_s-M_s)}\,dX_s$, $0 \leq t \leq T$, satisfies
$\eta^X=\eta^{\tilde{X}}$.  So the infimum in~\eqref{eq:75} coincides
with the infimum of this last expectation over all $\tilde{X} \in
\cX^2$. In fact, it coincides with its infimum over all increasing and bounded
$\tilde{X} \in \cX$:
 \begin{align}
   \label{eq:77}
&\inf_{X \in \cX^2} 
   \E_{\hat{\QQ}}\left[\Lambda^X_T-\hat{M}_TX_T\right]
\\&
 = \inf_{\tilde{X} \in \cX \text{ incr., bdd.}} \E_{\hat{\QQ}}\left[-\int_{[0,T]} |P_t-\hat{M}_t| \,d\tilde{X}_t+\frac{1}{2}\int_{[0,T]}
(\eta^{\tilde{X}}_t)^2 \mu(dt) \right]-\hat{M}_0x_0.
 \end{align}
It thus remains to show that this last infimum is
\begin{align}
  \label{eq:78}
  \inf_{\tilde{X} \in \cX \text{ incr., bdd.}} &\E_{\hat{\QQ}}\left[-\int_{[0,T]} |P_t-\hat{M}_t| \,d\tilde{X}_t+\frac{1}{2}\int_{[0,T]}
(\eta^{\tilde{X}}_t)^2 \mu(dt) \right]
\\&= -\frac{1}{2}\|\hat{\alpha}-\zeta_0\|^2_{L^2(\hat{\QQ}\otimes \mu)}+\frac{1}{2}\zeta_0^2\delta_0
\end{align}
for some $\hat{\alpha} \in L^2(\hat{\QQ} \otimes \mu)$.

We will argue below that there is a progressively measurable process
$a$ with upper-rightcontinuous paths such that $\sup_{\tau \leq v \leq
  .} a_v \in L^1(\hat{\QQ}\otimes\mu)$ with
\begin{align}
  \label{eq:79}
   |P_\tau-\hat{M}_\tau| \frac{\delta_\tau}{\rho_\tau} =
    \E_{\hat{\QQ}}\left[\int_{[\tau,T]} \sup_{\tau \leq v \leq u} a_v \,\mu(du)\middle | \cF_\tau
    \right]
\end{align}
for any stopping time $\tau \leq T$, i.e., such that the left-hand
side in~\eqref{eq:79} is the $\hat{\QQ}$-optional projection of the $\mu$-integral
on the right-hand side. Therefore, we get for any increasing and bounded
$\tilde{X} \in \cX$ that
\begin{align}
  \label{eq:80}
  \E_{\hat{\QQ}}&\left[-\int_{[0,T]} |P_t-\hat{M}_t| \,d\tilde{X}_t+\frac{1}{2}\int_{[0,T]}
(\eta^{\tilde{X}}_t)^2 \mu(dt)\right]
\\&=\E_{\hat{\QQ}}\left[-\int_{[0,T]} \int_{[t,T]} \sup_{t \leq v \leq
    u} a_v \,\mu(du) \,\frac{\rho_t}{\delta_t} \,d\tilde{X}_t+\frac{1}{2}\int_{[0,T]}
(\eta^{\tilde{X}}_u)^2 \mu(du)\right]
\\&=\E_{\hat{\QQ}}\left[\int_{[0,T]} \left\{
\frac{1}{2}(\eta^{\tilde{X}}_u)^2- \int_{[0,u]} \sup_{t \leq v \leq
    u} a_v \,d\eta^{\tilde{X}}_t \right\}\mu(du) \right], \label{eq:81}
\end{align}
where for the second equality we applied Fubini's theorem and used that by monotonicity of
$\tilde{X}$ and~\eqref{eq:39} we have
$\frac{\rho_t}{\delta_t}d\tilde{X}_t=d\eta^{\tilde{X}}_t$. Introducing
\begin{align}
  \label{eq:82}
  \hat{\alpha}_u \set \sup_{0 \leq v \leq u} a_v \vee \zeta_0, \quad 0 \leq u \leq T,
\end{align}
we can estimate the expression in $\{\dots\}$ in~\eqref{eq:81} by
\begin{align}
  \frac{1}{2}&(\eta^{\tilde{X}}_u)^2- \int_{[0,u]} \sup_{t \leq v \leq
    u} a_v \,d\eta^{\tilde{X}}_t
\\\label{eq:83} &\geq \frac{1}{2}(\eta^{\tilde{X}}_u)^2- \int_{[0,u]} \hat{\alpha}_u\,d\eta^{\tilde{X}}_t
= \frac{1}{2}(\eta^{\tilde{X}}_u)^2-
    \hat{\alpha}_u(\eta^{\tilde{X}}_u-\zeta_0)
\\&= \frac{1}{2}(\eta^{\tilde{X}}_u-\hat{\alpha}_u)^2-
   \frac{1}{2}(\hat{\alpha}_u-\zeta_0)^2+\frac{1}{2}\zeta_0^2
 \geq - \frac{1}{2}(\hat{\alpha}_u-\zeta_0)^2+\frac{1}{2}\zeta_0^2, \label{eq:84}
\end{align}
which does not depend on the choice of increasing, bounded $\tilde{X} \in
\cX$. Combining~\eqref{eq:81} with this estimate
thus gives
\begin{align}
   \label{eq:85}
\inf_{\tilde{X} \in \cX \text{ incr., bdd.}}& \E_{\hat{\QQ}}\left[-\int_{[0,T]} |P_t-\hat{M}_t| \,d\tilde{X}_t+\frac{1}{2}\int_{[0,T]}
(\eta^{\tilde{X}}_t)^2 \mu(dt) \right]
\\&\geq \E_{\hat{\QQ}}\left[\int_{[0,T]} \left\{-
                     \frac{1}{2}(\hat{\alpha}_u-\zeta_0)^2+\frac{1}{2}\zeta_0^2
                     \right\}\mu(du) \right]
\\&
=-\frac{1}{2}\|\hat{\alpha}-\zeta_0\|_{L^2(\hat{\QQ}\otimes\mu)}^2 +\frac{1}{2}\zeta_0^2\delta_0,
\end{align}
which proves ``$\geq$'' in our assertion~\eqref{eq:78}.

It remains to argue that, in fact, equality holds true, which in
particular includes showing
$\hat{\alpha} \in L^2(\hat{\QQ}\otimes \mu)$. We start by observing that
$\hat{\alpha}$ is at least in $L^1(\hat{\QQ} \otimes \mu)$ because
$\sup_{0 \leq v \leq .} a \in L^1(\hat{\QQ} \otimes \mu)$. Moreover,
$\hat{\alpha}$ is increasing from $\zeta_0$ and it is right-continuous
and adapted by the upper-rightcontinuity and progressive measurability
of $a$. We can thus consider the increasing $\hat{X} \in \cX$ with
$\eta^{\hat{X}}=\hat{\alpha}$. For $\tilde{X}=\hat{X}$ we clearly have
equality in~\eqref{eq:84}, and, in fact, also in~\eqref{eq:83}. Indeed, by
construction, $\hat{X}$ and thus $\eta^{\hat{X}}$ increase only at
times $t$ when our process $a$ reaches a new maximum beyond $\zeta_0$
so that
$\sup_{t \leq v \leq u} a_v =\sup_{0 \leq v \leq u} a_v =
\hat{\alpha}_u$ for any $u \geq t$ at these times. Now, with
$\tilde{X} = \hat{X}\wedge n$ in~\eqref{eq:81} we get from these
considerations that
\begin{align}
  \label{eq:86}
  \E_{\hat{\QQ}}&\left[-\int_{[0,T]} |P_t-\hat{M}_t| \,d (\hat{X}\wedge n)_t+\frac{1}{2}\int_{[0,T]}
(\eta^{\hat{X}\wedge n}_t)^2 \mu(dt)\right]
\\ &=\E_{\hat{\QQ}}\left[\int_{[0,T]} \left\{
\frac{1}{2}(\eta^{\hat{X}\wedge n}_u)^2- \int_{[0,u]} \sup_{t \leq v \leq
    u} a_v \,d\eta^{\hat{X}\wedge n}_t \right\}\mu(du) \right]
\\&=\int_{\{\hat{X}  \leq n\}}\left(-
    \frac{1}{2}(\hat{\alpha}-\zeta_0)^2+\frac{1}{2}\zeta_0^2\right)
    d(\hat{\QQ} \otimes \mu)
\\ \label{eq:87} &\quad+\int_{\{\hat{X}  > n\}}\left(
    \frac{1}{2}(\eta^{\hat{X}\wedge n})^2-\hat{\alpha}\eta^{\hat{X}\wedge n}+\hat{\alpha}\zeta_0\right) d(\hat{\QQ} \otimes \mu).
\end{align}
Once we know that
$\hat{\alpha} = \eta^{\hat{X}} \geq \eta^{\hat{X}\wedge n}$ is in
$L^2(\hat{\QQ}\otimes\mu)$, we can use, respectively, monotone and
dominated convergence to let $n \uparrow \infty$ in the preceding expression
and conclude that
\begin{align}
  \label{eq:88}
  \inf_{\tilde{X} \in \cX \text{ incr., bdd.}} &\E_{\hat{\QQ}}\left[-\int_{[0,T]} |P_t-\hat{M}_t| \,d\tilde{X}_t+\frac{1}{2}\int_{[0,T]}
                                                 (\eta^{\tilde{X}}_t)^2 \mu(dt)\right]
  \\& \leq  \int_{\Omega \times [0,T]}\left(-
      \frac{1}{2}(\hat{\alpha}-\zeta_0)^2+\frac{1}{2}\zeta_0^2\right)
      d(\hat{\QQ} \otimes \mu)+0
  \\&=-\frac{1}{2}\|\hat{\alpha}-\zeta_0\|_{L^2(\hat{\QQ}\otimes\mu)}^2
      +\frac{1}{2}\zeta_0^2\delta_0
\end{align}
as remained to be shown for our claim~\eqref{eq:78}. Now, use the
estimate
\begin{align}
  \label{eq:89}
  \E_{\hat{\QQ}}&\left[-\int_{[0,T]} |P_t-\hat{M}_t| \,d (\hat{X}\wedge n)_t+\frac{1}{2}\int_{[0,T]}
(\eta^{\hat{X}\wedge n}_t)^2 \mu(dt)\right]
\\ &\geq -\left\|\sup_{0 \leq t \leq
  T}\left(|P_t-M_t|\frac{\delta_t}{\rho_t}\right)\right\|_{L^2(\hat{\QQ})}\|
  \eta^{\hat{X}\wedge n}_T-\zeta_0\|_{L^2(\hat{\QQ})}+\frac{1}{2}\|
  \eta^{\hat{X}\wedge n}\|_{L^2(\hat{\QQ}\otimes \mu)}^2
\end{align}
to see that if $\hat{\alpha}=\eta^{\hat{X}}$ was not in
$L^2(\hat{\QQ}\otimes\mu)$ then the expectation in~\eqref{eq:86} would
tend to $+\infty$ by monotone convergence as $n \uparrow \infty$. At
the same time, though, the first integral in~\eqref{eq:87} would
converge to $-\infty$. Moreover,
$\hat{\alpha} \in L^1(\hat{\QQ}\otimes\mu)$ ensures that the
contribution of $\hat{\alpha}\zeta_0$ to the second
$\hat{\QQ}\otimes \mu$-integral there vanishes for
$n\uparrow\infty$. By choice of $\hat{X}$, we have
$\hat{\alpha} = \eta^{\hat{X}} \geq \eta^{\hat{X}\wedge n}$, so that
the remaining contribution from this integral is less than or equal to
0. Hence, the assumption
$\hat{\alpha} \not\in L^2(\hat{\QQ}\otimes\mu)$ leads us to the
contradiction that the identical quantities in~\eqref{eq:86}
and~\eqref{eq:87} would converge to $+\infty$ and $-\infty$ at the
same time when $n\uparrow \infty$.

For the completion of our proof, we still need to construct the
process $a$ from~\eqref{eq:79}. It will be obtained by the
representation theorem from~\cite{BankElKaroui:04}. For this we note
that, while having full support on $[0,T]$ by Assumption~\ref{asp:2},
our measure $\mu$ is not directly applicable for this representation
theorem since it has an atom at time $T$. We thus replace it with the
atomless optional random measure
$\tilde{\mu}(dt) = 1_{[0,T)}(t)\mu(dt) + \lambda e^{-\lambda (t-T)}1_{[T,\infty)}(t)dt$ on
$[0,\infty)$ where $\lambda \set \mu(\{T\})$. We also extend $Y_t\set
|P_t-\hat{M}_t| \frac{\delta_t}{\rho_t}$, $0 \leq t \leq T$, to a
process on $[0,\infty)$ by letting $Y_t \set Y_Te^{-\lambda(t-T)}$ for
$t \geq T$ and we let $\cF_t \set \cF_T$ for $t \geq T$. Then, by
Assumption~\ref{asp:1}, the process $Y$ is adapted, continuous with
limit $\lim_{t \uparrow \infty} Y_t=0$ and it is of class (D) since it
has an integrable upper bound because of $M \in \cM^2(\hat{\QQ})$
and~\eqref{eq:59}. We thus can apply Theorem~3 of
\cite{BankElKaroui:04} in connection with their Remark~2.1 to obtain
an upper-right continuous, progressively measurable $a$ such that  for any stopping time
$\tau$ we have $\sup_{\tau \leq v \leq
  .} a_v \in L^1(\hat{\QQ}\otimes \tilde{\mu})$ with
\begin{align}
  \label{eq:90}
  Y_\tau &= \E_{\hat{\QQ}}\left[\int_{[\tau,\infty)} \sup_{\tau \leq v
  \leq u} a_v \, \tilde{\mu}(du)\middle| \cF_\tau\right].
\end{align}
In fact, for $t \geq T$, one readily checks that $a_t = a_T=Y_T$ will
do. Therefore, we get by uniqueness of $a$ that for any stopping
time $\tau \leq T$ the above representation amounts to
\begin{align}
  \label{eq:91}
     |P_\tau-\hat{M}_\tau| \frac{\delta_\tau}{\rho_\tau} = Y_\tau&=
    \E_{\hat{\QQ}}\left[\int_{[\tau,T)} \sup_{\tau \leq v \leq u} a_v
  \,\tilde{\mu}(du)+ \int_{[T,\infty)} \sup_{\tau \leq v \leq T} a_v \tilde{\mu}(dt)\middle | \cF_\tau
    \right]\\
&= \E_{\hat{\QQ}}\left[\int_{[\tau,T)} \sup_{\tau \leq v \leq u} a_v
  \,\mu(du)+ \sup_{\tau \leq v \leq T} a_v \mu(\{T\}) \middle | \cF_\tau
    \right]
\end{align}
as requested.
\end{proof}

The proof of the upper bound in our duality~\eqref{eq:11} of
Theorem~\ref{thm:1} is now easy to complete. Indeed, the constructed
triple $(\hat{\QQ},\hat{M},\hat{\alpha})$ is as requested by our
theorem. Moreover, recalling the definition~\eqref{eq:61} of
$\hat{v}_0$ and combining~\eqref{eq:69} with \eqref{eq:75} gives the
desired upper bound~\eqref{eq:58}.

\bibliographystyle{plainnat} \bibliography{finance}

\end{document}